\documentclass[final]{IEEEtran}
\usepackage{setspace}
\usepackage{color}
\usepackage{csquotes}
\usepackage{enumitem}
\usepackage{float}

\floatstyle{ruled}
\newfloat{algorithm}{tbp}{loa}
\providecommand{\algorithmname}{Algorithm}
\floatname{algorithm}{\protect\algorithmname}
\usepackage{hyperref}
\usepackage{amsopn}

\makeatletter
\newcommand{\manuallabel}[2]{\def\@currentlabel{#2}\label{#1}}
\makeatother
\usepackage{amsmath}
\usepackage{amssymb}
\usepackage{amsthm}
\usepackage{bbm}
\usepackage{tikz}
\usepackage{mathrsfs}
\usepackage{pgfplots}
\pgfplotsset{compat=1.14}
\usetikzlibrary{arrows}
\usetikzlibrary{arrows.meta}
\newtheorem{lemma}{Lemma}
\newtheorem{theorem}{Theorem}
\newtheorem{proposition}{Proposition}

\newtheorem{cor}{Corollary}

\title{Simple Coding Techniques \\ for Many-Hop Relaying}

\author{Yan Hao Ling and Jonathan Scarlett\thanks{The authors are with the  Department of Computer Science, National University of Singapore (NUS). Jonathan Scarlett is also with the Department of Mathematics, NUS, and the Institute of Data Science, NUS.  (e-mail: \url{e0174827@u.nus.edu};  \url{scarlett@comp.nus.edu.sg})}\thanks{This work was supported by the Singapore National Research Foundation (NRF) under grant number R-252-000-A74-281.}\thanks{\copyright 2022 IEEE.  This work is published IEEE Transactions on Information Theory, and can be found at \url{https://doi.org/10.1109/TIT.2022.3180001}.}}

\newcommand{\Psf}{\mathsf{P}}
\newcommand{\calO}{\mathcal{O}}
\newcommand{\epsbar}{\bar{\epsilon}}

\begin{document}
    \maketitle
    \begin{abstract}
        In this paper, we study the problem of relaying a single bit of information across a series of binary symmetric channels, and the associated trade-off between the number of hops $m$, the transmission time $n$, and the error probability.  We introduce a simple, efficient, and deterministic protocol that attains positive information velocity (i.e., a non-vanishing ratio $\frac{m}{n}$ and small error probability) and is significantly simpler than existing protocols that do so.  In addition, we characterize the optimal low-noise and high-noise scaling laws of the information velocity, and we adapt our 1-bit protocol to transmit $k$ bits over $m$ hops with $\calO(m+k)$ transmission time. 
    \end{abstract}

    \begin{IEEEkeywords}
        Information velocity, relay channels, many-hop relaying, network information theory
    \end{IEEEkeywords}
    
    \medskip
    \section{Introduction}
    
    Relay channels are one of the most fundamental building blocks of network information theory.  While extensive research attention has been paid to few-hop (particularly two-hop) relaying settings, the information-theoretic understanding of many-hop systems (e.g., a number of hops growing with the block length) is much more limited.  In this paper, we seek to improve the understanding of the simple yet fundamental problem of relaying one bit (or more generally, multiple bits) over a long tandem of channels (e.g., see \cite{onebit,schulman_1994}; we discuss the history of this problem in Section \ref{sec:related}).
    
    \subsection{Problem Setup}
    
    We consider a chain of $m+1$ nodes indexed by $\{0,1,\dotsc,m\}$; node $0$ is the encoder, node $m$ is the decoder, and the remaining nodes are relays.  Node $0$ is given a binary random variable $\Theta$ that takes a value in $\{0,1\}$, each with probability $\frac12$.  At each time step $i \in \{1,2,\ldots, n\}$, node $j$ ($0 \leq j \leq m-1$) sends one bit $X_{i,j}$ to node $j+1$ through a binary symmetric channel (BSC) with crossover probability $p \in \big(0,\frac{1}{2}\big)$; we denote the output of this transmission by $Y_{i,j}$.  The channels are assumed to be memoryless and independent across nodes. Importantly, $X_{i,j}$ can only depend on what was previously received, i.e., $Y_{1,j-1}, Y_{2,j-1}, \ldots, Y_{i-1, j-1}$.  After $n$ time steps, node $m$ outputs an estimate $\hat{\Theta}$ of $\Theta$.
    
    For a given number of nodes $m$ and BSC crossover probability $p$, let $n^*(m,p)$ be the minimum block length $n$ for which there exists a protocol (encoding, relaying, and decoding) such that $\mathbb{P}(\hat{\Theta} \ne \Theta) \le \frac{1}{3}$.\footnote{Any other fixed constant in $\big(0,\frac{1}{2}\big)$ could also be used, or one could require that $\mathbb{P}(\hat{\Theta} \ne \Theta) \to 0$ as $m \to \infty$; our results will handle all of these variants.}  The {\em information velocity} is defined as follows:\footnote{We are not aware of a formal proof that this limit always exists.  However, all of our results will remain valid regardless of whether the $\liminf$ or $\limsup$ is used in place of the regular limit.}
    \begin{equation}
        v(p) = \lim_{m \rightarrow \infty} \frac{m}{n^*(m,p)}. \label{eq:iv}
    \end{equation}
    We note that $v(0) = 1$, and $v(p) \in [0,1]$ for all $p$.  Analogously, we say that a protocol (or more precisely, a sequence of protocols indexed by $m$) {\em attains positive information velocity} if\footnote{Throughout the paper, we use the standard asymptotic notation $\calO(\cdot)$, $\Omega(\cdot)$, $\Theta(\cdot)$, $o(\cdot)$, and $\omega(\cdot)$.  The distinction between asymptotic $\Theta(\cdot)$ notation and the unknown bit $\Theta \in \{0,1\}$ will be clear from the context, as well as the presence/absence of brackets.} $n = \calO(m)$ and the error probability is at most $\frac{1}{3}$ for all sufficiently large $m$.  We will also be interested in the possible decay rate of the error probability towards zero.
    
    \subsection{Contributions and Related Work} \label{sec:related}
    
    {\bf Contributions.} In this paper, we present a simple, efficient, and deterministic protocol attaining positive information velocity (see Theorem \ref{thm:main_1bit} and Corollary \ref{cor:main_1bit} below), based on the idea of having different nodes estimate $\Theta$ by processing bit sequences of different lengths (e.g., most nodes process 1 or 3 bits at a time, but a small minority of nodes consider blocks of size $\omega(1)$ or even $\calO(m)$) and using a recursive decoding technique.  In addition, we characterize the asymptotic behavior of $v(p)$ in the limits $p \to 0$ and $p \to \frac{1}{2}$, showing that $v(p) = 1-\Theta(p)$ as $p \to 0$, and $v\big(\frac{1-\delta}{2}\big) = \Theta(\delta^2)$ as $\delta \to 0$.  Finally, we provide a multi-bit generalization of our protocol that can reliably transmit $\calO(k)$ bits in $\calO(m+k)$ transmissions.
    
    \begin{figure*}
    \begin{centering}
        \includegraphics[width=0.8\textwidth]{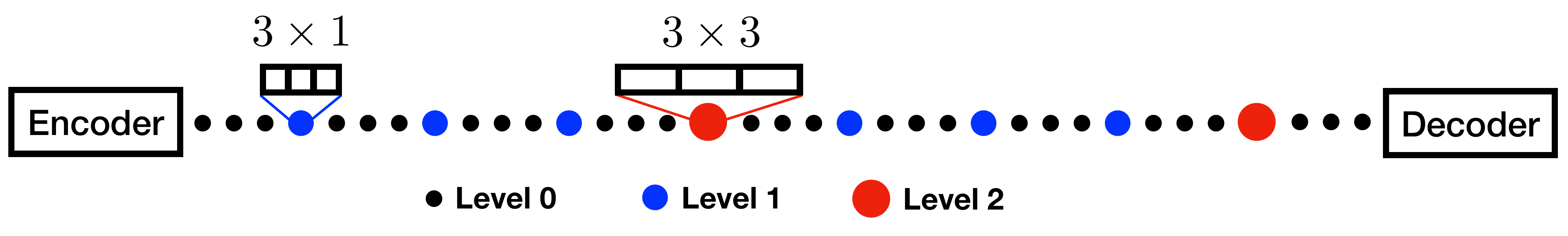}
        \par
    \end{centering}
    
    \caption{Illustration of our protocol; level 0 nodes directly transmit their input, whereas higher level nodes process blocks. \label{fig:OneBit}}
\end{figure*}
    
    {\bf Related work.} As noted in \cite{onebit}, the preceding problem was posed at the 2015 Simons Institute program on information theory by Yury Polyanskiy, who also coined the term ``information velocity''.  At the same workshop, Leonard J.~Schulman pointed out an existing solution based on tree codes \cite {schulman_1994}. 
    
    In more detail, Rajagopalan and Schulman \cite{schulman_1994} (building on earlier work by Schulman \cite{schulman1992comm,schulman1993det}) studied general techniques for relating noiseless vs.~noisy computation in distributed systems using tree codes.  An important implication of their main result is that $v(p) > 0$, which is obtained by specializing to a chain-structured network.  However, the associated computational complexity is exponential, requiring the use of recursive maximum-likelihood decoding.  A closed-form upper bound (i.e., converse) on the information velocity can also be inferred directly from \cite{schulman_1994}, which we will strengthen as part of our analysis of the high-noise asymptotics.  In the low-noise limit, we find that the converse of \cite{schulman_1994} already provides tight asymptotic scaling. 
    
    A simple and efficient protocol for the setting we consider (and its multi-bit extension) was also proposed in \cite{schulman_1994}, but it falls {\em very narrowly} short of attaining positive information velocity, due to requiring $n = \omega(m)$ transmissions; the $\omega(m)$ scaling is extremely mildly super-linear according to the iterated logarithm.  This $\omega(m)$ behavior arises from each node waiting to receive $\omega(1)$ useful bits from the previous node before forwarding (see Appendix \ref{sec:schulman} for more details), which our protocol avoids.
    
    A more recent work of Gelles {\em et al.}~\cite{gelles2011efficient} introduces a variant of tree codes that can be decoded in polynomial time.  The focus in \cite{gelles2011efficient} is on two-way communication channels, but based on \cite{schulman_1994}, it seems feasible that their techniques could be adapted to one-way settings such as the one we consider.  On the other hand, such an adaptation would still leave the following relative limitations compared to our approach:
    \begin{itemize}
        \item According to \cite[Thm.~V.I]{gelles2011efficient}, when $n = \calO(m)$, the error probability can (only) be controlled to be a constant such as $\frac{1}{3}$, whereas our protocol will give stronger $e^{-\Theta(\sqrt{n})}$ decay, and simple variations will further boost this to $e^{-\Theta(n^{1-\epsilon})}$ for any $\epsilon > 0$.  On the other hand, we still fall short of the $e^{-\Theta(n)}$ decay shown using the inefficient approach of \cite{schulman_1994}.
        \item The computational complexity in \cite{gelles2011efficient} in ${\rm poly}(n)$ with an unspecified power, whereas in our protocol the nodes each incur a computational complexity of $\calO(n)$.
        \item The approach in \cite{gelles2011efficient} is randomized, and assumes the availability of common randomness across the nodes, whereas our protocol is deterministic.
    \end{itemize}
    Perhaps more important than any of these factors is the fact that that the tree codes of \cite{schulman_1994,gelles2011efficient} were designed for much more general distributed computation settings, and are accordingly highly versatile but unnecessarily complex for the simpler information velocity problem.  In our judgment, the main advantage of our protocol is its simplicity.
    
    A recent work (concurrent with ours) established the exact information velocity of the binary erasure channel (BEC) and other erasure-type channels \cite{Dom22}, and also studied the exponential decay rate of the error probability.  The protocols therein appear to rely heavily on the property $\Pr[x~{\rm sent} \,|\, x~{\rm received}]=1$ for every non-zero input $x$ (e.g., this holds for the BEC and the Z-channel), which allows information to be propagated with certainty.  As a result, handling the BSC with their approach does not appear to be feasible.
    
    Another related line of works studies the optimal error exponent for transmitting one bit over a 2-hop relay network \cite{onebit,jog2020teaching,Lin21}.  To our knowledge, none of the coding techniques in these works are suited to attaining positive information velocity.  However, we note that the work of Huleihel, Polyanskiy, and Shayevitz \cite{onebit} highlighted interesting connections between the information velocity problem and the 2-hop setting.  In particular, based on the observation that positive information velocity is possible, they conjectured that the 1-hop and 2-hop error exponents should coincide in the high-noise limit so that ``information propagation does not slow down''.  This conjecture was resolved in the affirmative in our subsequent work \cite{Lin21}, where we showed that the 1-hop and 2-hop exponents in fact coincide at \emph{any} noise level.
    
    Finally, we briefly mention that there exists extensive work on relay channels studying the capacity \cite[Ch.~16]{el2011network}, error exponent \cite{tan2015reliability}, hypothesis testing protocols \cite{salehkalaibar2019hypothesis}, and so on.  We also briefly mention that related many-hop transmission problems were studied for (continuous) channels with input constraints in \cite{polyanskiy2015dissipation}.

    \section{The Protocol and its Guarantees} \label{sec:protocol}
    
    \subsection{Protocol and Main Result}
    
    We will first describe the protocol assuming $p$ to be sufficiently small (namely, below $\frac{1}{48}$), and then turn to arbitrary $p \in \big(0,\frac{1}{2}\big)$.
    
    The protocol is illustrated in Figure \ref{fig:OneBit}, and described in detail as follows.  We arrange the nodes into levels. A (relay) node indexed by $i \in \{1,\dotsc,m-1\}$ has level $l \in \{0,1,2,\dotsc\}$ if $4^l$ divides $i$ but $4^{l+1}$ does not.  The nodes at various levels operate as follows:
    \begin{itemize}
        \item Node 0 (the encoder) repeatedly transmits $\Theta$ throughout the course of the protocol.
        \item A level-$0$ node reads one bit, ``estimates'' $\Theta$ to be equal to that bit, and forwards its estimate in one channel use.  This is done repeatedly throughout the course of the protocol. 
        \item A level-$1$ node reads 3 bits, estimates $\Theta$ according to a majority vote, and forwards its estimate in 3 channel uses (i.e., repeating 3 times).  Again, this is done repeatedly.
        \item Continuing recursively, a level-$l$ node reads in $3^l$ bits, breaks them up into 3 blocks of $3^{l-1}$ bits each, performs level-($l-1$) decoding on each block, and then estimates $\Theta$ according to a majority of 3.  It then forwards its estimate by repeating $3^l$ times.
        \item Node $m$ (i.e., the decoder) applies level-$L$ decoding with $L = \lfloor \log_4 m \rfloor$, regardless of whether or not $m$ is a power of $4^l$ for any $l$.  Accordingly, the encoder sends out $3^L$ bits.
    \end{itemize}
    The intuition behind this protocol is that while some coding is needed to ensure that the information is not lost, not every node needs to code over a long length (which would cause significant delays).  Instead, reliability can be maintained with most nodes doing little or no coding, but a small fraction of nodes doing more.  We note that the constants $3$ and $4$ can be replaced by other values (subject to the latter being larger), as we discuss in Appendix \ref{app:error_prob}.
    
    We observe that the propagation of information is delayed in the protocol (e.g., even in a simpler setting where every node adopts the level-0 strategy, node $i$ would not receive meaningful information until time $i$).  To account for this, each node ignores all initial bits whose values are unspecified above,\footnote{This is possible because the protocol is assumed to be known to all nodes.} and a level-$l$ relay node only starts transmitting after receiving $3^l$ (non-ignored) bits.  
    
    
    We are now ready to state our first main result, which will be proved in the following subsection.
    
    \begin{theorem} \label{thm:main_1bit}
        If $p \le \frac{1}{48}$, then for any number of hops $m$, the preceding protocol attains an error probability of at most $\frac{1}{12}$ with a total transmission time of $n = \calO(m)$.
    \end{theorem}
    
    In Corollary \ref{cor:main_1bit} below, we state a variant of this result that holds for all $p \in \big(0,\frac{1}{2}\big)$ and gives decaying error probability.  Moreover, in Section \ref{sec:1bit_var} below, we discuss explicit bounds on $v(p)$, as opposed to the fact that $v(p) > 0$ alone (i.e., $n = \calO(m)$ with an unspecified constant).
    
    \subsection{Analysis of the Protocol} \label{sec:1bit_analysis}
    
    Here we prove Theorem \ref{thm:main_1bit}.  Let $\epsilon_l$ be the error probability associated with sending a single level-$l$ block over a chain of length $4^l$ between two consecutive level-$l$ nodes (or from the encoder to the first one).
    
    \begin{lemma} \label{lem:pe_recursion}
        Under the preceding protocol, for any $l \ge 0$, if $\epsilon_l \leq \frac{1}{48}$, then $\epsilon_{l+1} \leq \frac{1}{48}$.
    \end{lemma}
    \begin{proof}
        Consider a block of $3^{l+1}$ bits corresponding to a level-$(l+1)$ block.  According to the protocol, this amounts to 3 parallel level-$l$ blocks, each of which passes through 4 level-$l$ hops (each such `hop' being a length-$4^l$ chain).  By the union bound, the failure probability for each of these level-$l$ blocks is at most $4 \epsilon_{l}$ (if a block fails, it must have failed at some hop).
        
        Observe that if we take majority of 3, and each of them independently has a failure rate of $q$, then the overall failure rate is at most $3q^2$ (i.e., at least two of the three must fail).  Since $q \le 4\epsilon_l$, the failure rate of the entire level-($l+1$) block is at most $3\cdot (4\epsilon_l)^2 = 48\epsilon_l^2 \leq \frac{1}{48}$, as required.
    \end{proof}
    
    Since $\epsilon_0 = p$, it follows by induction that if $p \leq \frac{1}{48}$, then the failure rate for any single block at any level is at most $\frac{1}{48}$.  
    
    In the case that $m$ is a power of 4, the overall error probability is precisely $\epsilon_L$ with $L = \log_4 m$, so is at most $\frac{1}{48}$.  On the other hand, if $m$ is not a power of $4$, we first consider a decoder that differs from the one we described, and then show how the analysis transfers from the former to the latter.  Specifically, we imagine a decoder that internally simulates extra hops (including artificial noise) with the protocol's level structure up to the next multiple of $4^L$, where $L = \lfloor \log_4 m \rfloor$ (e.g., increasing to $3 \times 4^2$ hops in Figure \ref{fig:OneBit}), concluding with level-$L$ decoding.  In this case, the number of invocations of level-$L$ decoding is $\lceil \frac{m}{4^L}\rceil \le 4$ (including at the final node), so by the union bound, the error probability is at most $4 \cdot \frac{1}{48} = \frac{1}{12}$. 

    The decoder that we consider has no such internal simulation, and instead directly applies level-$L$ decoding.  To see that the same result applies in this scenario, we first note that our analysis (in particular, Lemma \ref{lem:pe_recursion}) still applies whenever each BSC has crossover probability {\em at most} $p$, rather than {\em exactly} $p$.  Hence, in the above-mentioned internal simulation of extra hops, the noise can be removed without affecting the result.  Once this is done, due to the recursive nature of our decoder, running level-$L$ decoding directly is no different from running it following the lower-level simulated nodes with noiseless links (note that applying a level-$l$ decoder twice in succession is mathematically the same as applying it once).  Hence, the above-established bound of $\frac{1}{12}$ again applies.
    
    
    To attain positive information velocity, we need to check that the total transmission time $n$ is at most linear in the number of hops.  We can break down the total transmission time into two parts:
    \begin{itemize}
        \item {\bf Transmission delay}, which is the number of bits that our protocol specifies the encoder to send out.
        \item {\bf Propagation delay}, which is the time lapse between the encoder transmitting its first bit and the last node receiving that bit (or `first' could be replaced by any other bit).
    \end{itemize}
    We already showed that the transmission delay is $3^L$ with $L = \lfloor \log_4 m \rfloor$, and is thus sub-linear with respect to $m$.  It remains to consider the propagation delay.
    
    A level-$l$ relay node waits for $3^l$ bits to be received before it starts transmitting.  Hence, the total wait time summed over the $m-1$ relay nodes is at most
    \begin{equation}
        \sum_{l} m\cdot \frac{3^l}{4^l} \leq 4m,
    \end{equation}
    since by construction there are at most $\frac{m}{4^{l}}$ relay nodes at level $l$.  This is linear in $m$, as desired, completing the proof of Theorem \ref{thm:main_1bit}.
    
    \subsection{Variations and Discussion} \label{sec:1bit_var}
    
    Two notable limitations of Theorem \ref{thm:main_1bit} are the assumption $p \le \frac{1}{48}$ and the fact that the error probability is only upper bounded by $\frac{1}{12}$.  We proceed by showing that minor variations of the protocol and analysis can overcome these limitations.  We then discuss the computational complexity and explicit bounds on $v(p)$.
    
    {\bf General noise levels.} If the noise level $p$ is larger than $\frac{1}{48}$, we can use a constant-length repetition code to reduce the effective crossover probability below $\frac{1}{48}$, which amounts to a constant factor blowup in the number of channel uses.  That is, every node always waits for $r \ge 1$ bits, treats the majority vote over these $r$ bits as a single bit, and applies the above protocol.  As long as $p < \frac{1}{2}$, a large enough choice of $r$ will bring the effective crossover probability below $\frac{1}{48}$.  Hence, we immediately deduce positive information velocity for all $p \in \big(0,\frac{1}{2}\big)$.  A refined variant of this idea is explored in Section \ref{sec:noise_asymp}.
    
    {\bf Improved error probability.} Suppose that we assume $p \le \frac{1}{96}$ instead of $p \le \frac{1}{48}$ (again noting that we can generalize using the idea of repetition).  The proof of Lemma \ref{lem:pe_recursion} reveals that $\epsilon_{l+1} \le 48 \epsilon_{l}^2$, so with $\epsilon_0 \le \frac{1}{96}$, we get $\epsilon_1 \le \frac{1}{2} \cdot \frac{1}{96}$, followed by the error probability being squared (up to a constant factor) with each level increase.  To understand the behavior after $L$ levels, we momentarily ignore the factor of $48$ in the recursion $\epsilon_{l+1} \le 48 \epsilon_{l}^2$ (which will not impact the final result), and accordingly consider the quantity $p_L = {\rm sq}({\rm sq}(\dotsc {\rm sq}(p)))$, where the square function ${\rm sq}(\cdot)$ is applied $L$ times.  Taking the log gives $\log p_L = 2^{L} \log p$, and since $L = \log_4 m + \calO(1) = \frac{1}{2} \log_2 m + \calO(1)$, it follows that $p_L = e^{-\Theta(\sqrt{m})}$ for fixed $p$.  A simple generalization of this argument (accounting for the factor of 48) reveals that the error probability of our protocol decays as $e^{-\Theta(\sqrt{m})}$, or equivalently $e^{-\Theta(\sqrt{n})}$. 
    
    The above discussion is summarized in the following corollary.
    
    \begin{cor} \label{cor:main_1bit}
        For any fixed $p \in \big(0,\frac{1}{2}\big)$, our protocol, combined with a suitably-chosen constant number of repetitions, attains error probability at most $e^{-\Theta(\sqrt{n})}$ with a total transmission time of $n = \calO(m)$.
    \end{cor}
    
    In Appendix \ref{app:error_prob}, we show that the error probability can be improved to $e^{-\Theta(n^{1-\epsilon})}$ (for any $\epsilon > 0$) by adjusting the constants in our protocol and/or chaining multiple instances of the protocol one after the other.
    
    {\bf Computational complexity.} Naively, when a node processes a block of $B$ bits (namely, $B = 3^l$ at level $l$), it spends $\calO(B)$ computation time forming its estimate of $\Theta$ at the final time step of that block, e.g., by forming a ternary tree with the received bits as the leaves, and propagating up to the root.  However, this computation can be reduced by starting the processing before the whole block arrives: The first 3 bits are combined using a majority vote (and then the original 3 are discarded), the first 3 blocks-of-3 are combined when available, and so on.  In this manner, a level-$l$ node requires at most $\calO(l) \le \calO(\log m)$ computation in any given time step (i.e., the depth of the above-mentioned ternary tree), while still maintaining $\calO(B)$ total computation across the whole block of length $B$.  In particular, each node incurs $\calO(n)$ computation time in total (i.e., linear time).
    
    {\bf Explicit bounds on $v(p)$.} Throughout the paper, we provide results stating that $n = \calO(m)$, thus establishing that $v(p) > 0$, but without specifying the precise constants.  Nevertheless, the implied constants can be inferred from the proofs, with varying degrees of technical difficulty.  For example, we have the following:
    \begin{itemize}
        \item The analysis of Theorem \ref{thm:main_1bit} demonstrates that $v\big(\frac{1}{48}\big) \ge \frac{1}{4}$.
        \item Using Hoeffding's inequality (stated in \eqref{eq:Hoeffding} below) in our handling of general noise levels above, it is straightforward to verify that for $p > \frac{1}{48}$, it holds that $v(p) \ge \big( 4 \big\lceil \frac{2 \log 48}{(1-2p)^2} \big\rceil  \big)^{-1}$.  When the ceiling function equals (or is sufficiently close to) its argument, the right-hand side can further be lower bounded by $\frac{1}{31} (1-2p)^2$.
        \item In Proposition \ref{prop:converse} below, we will see that $v(p) \le (1-2p)^2$ for all $p$. 
        \item In Section \ref{sec:noise_asymp}, we will show that $v(p) \to 1$ as $p \to 0$, but we do not attempt to establish how close we get for any fixed $p$; the required $p$ to establish $v(p) \ge 0.9$ (say) appears to be very small using our methods.
    \end{itemize}
    From the second and third dot points above, we see that the achievability and converse bounds can differ by a large constant factor.  A refined analysis of our protocol (or variants thereof) may reduce this constant, e.g., replacing the values 3 and 4 in our protocol by other choices (see Appendix \ref{app:error_prob}), or using a tighter concentration bound than Hoeffding's inequality.  However, overall, attaining near-matching bounds for all $p$ remains an open challenge.

    \section{Low-Noise and High-Noise Asymptotics} \label{sec:noise_asymp}
    
    In this section, we study the asymptotics of $v(p)$ in the low-noise limit $p \to 0$, and the high-noise limit $p \to \frac{1}{2}$.  We start with the former, which turns out to be more straightforward.
    
    \subsection{Low Noise}
    
    Here we prove the following.
    
    \begin{theorem} \label{thm:low_noise}
        The information velocity satisfies $v(p) = 1 - \Theta(p)$ as $p \to 0$.
    \end{theorem}
    
    This scaling can be contrasted with the channel capacity of the BSC, which scales as $1 - \Theta\big(p \log \frac{1}{p}\big)$ as $p \to 0$.  Hence, at least in some cases, the information velocity can be strictly higher than the capacity. 
    
    The converse part of Theorem \ref{thm:low_noise} follows immediately from \cite[Thm.~7.1]{schulman_1994} (see also Proposition \ref{prop:converse} below), so it remains to prove the achievability part.  In the form given, Theorem \ref{thm:main_1bit} gives a bound on $v(p)$ that fails to decrease as we decrease $p$ below $\frac{1}{48}$.  However, we can capture the dependence on $p$ by simply letting a higher fraction of nodes (depending on $p$) have level $0$.  More specifically, we fix $c \ge 1$ and re-assign the levels as follows:
    \begin{itemize}
        \item A node has level $l$ if its index is a multiple of $c\cdot 4^l$ but not a multiple of $c \cdot 4^{l+1}$. If there is no such $l$, we set the level to 0.
        \item Apart from changing the levels, the rest of the protocol remains unchanged.
    \end{itemize}
    We proceed to bound the level-$1$ error probability $\epsilon_1$. The distance between consecutive level-1 nodes is $4c$, so each bit sent from a level-1 node to the next one is flipped with probability at most $4pc$.
    Since the level-1 nodes take a majority of 3, the error rate satisfies $\epsilon_1 \le 3(4pc)^2$, and by choosing $c = \big \lfloor \frac{1}{48p}\big\rfloor$, we obtain $\epsilon_1 \leq \frac{1}{48}$. Since the inequality $\epsilon_{l+1} \leq 3\cdot (4\epsilon_l)^2$ still holds (with the same proof as before), we can again obtain Lemma \ref{lem:pe_recursion}, and recursively apply it for all of the higher levels.
    
    To calculate the total transmission time $n(m,p)$, observe that for each $l\geq 1$, there are at most $\frac{m}{c4^l}$ nodes of level $l$. Therefore,
    \begin{multline}
        n(m,p) \leq m + \sum_{l\geq1} \frac{m}{c\cdot 4^l} \cdot 3^l \leq m + \frac{3m}{c} = m\Big(1+\frac{3}{c}\Big) \\
            = m(1+\Theta(p)),
    \end{multline}
    where we used the fact that $c = \Theta\big(\frac{1}{p}\big)$.  It then follows from the definition of $v(p)$ in \eqref{eq:iv} that $v(p) \geq 1 - \Theta(p)$ as $p \to 0$, completing the proof of Theorem \ref{thm:low_noise}.

    
    \subsection{High Noise}
    
    To study the high-noise asymptotics, we consider $p = \frac{1-\delta}{2}$ (i.e., $\delta = 1-2p$) for small $\delta > 0$.
    
    \begin{theorem} \label{thm:high_noise}
        The information velocity satisfies $v\big(\frac{1-\delta}{2}\big) = \Theta(\delta^2)$ as $\delta \to 0$.
    \end{theorem}
    
    We observe that in this case, the information velocity and the channel capacity exhibit the same high-noise asymptotics; the latter is also easily shown to be $\Theta(\delta^2)$.
    
    As noted in Section \ref{sec:1bit_var}, we can use a repetition code to bring down the effective crossover probability below $\frac{1}{48}$.  To make this more precise, we consider Hoeffding's inequality \cite[Ch.~2]{Bou13}, which gives that for i.i.d.~random variables $Z_1,\dotsc,Z_N$ distributed as ${\rm Bernoulli}(p)$, we have
    \begin{equation}
        \Pr\bigg( \sum_{i=1}^N (Z_i - p) \ge N\epsilon \bigg) \le e^{-2N\epsilon^2}. \label{eq:Hoeffding}
    \end{equation}
    In our setting, the relevant choice is $\epsilon = \frac{1}{2} - p = \frac{\delta}{2}$, since we are interested in whether $\sum_{i=1}^N Z_i$ exceeds $\frac{N}{2}$.  We find that to make the right-hand side of \eqref{eq:Hoeffding} at most $\frac{1}{48}$, we can set the number of repetitions to $N = \Theta(\epsilon^{-2}) = \Theta(\delta^{-2})$.  Thus, the block length is blown up by a factor of $\Theta(\delta^{-2})$, and the achievability part of Theorem \ref{thm:high_noise} follows.
    
    While the converse in \cite[Thm.~7.1]{schulman_1994} can also be applied to the high-noise setting, it only gives a suboptimal converse with $\Theta(\delta)$ dependence.  We proceed by giving a tighter analysis attaining the desired $\Theta(\delta^2)$ dependence, which we formally state as follows.
    
    \begin{proposition} \label{prop:converse}
        For any $\delta \in (0,1)$, we have $v\big(\frac{1-\delta}{2}\big) \leq \delta^2$.
    \end{proposition}
    
    To prove this proposition, we make use of the {\em strong data processing inequality}, which is now well-known (e.g., see \cite{ahlswede1976spreading,polyanskiy2017strong}) and is stated as follows.
    
    \begin{lemma} \label{lem:sdpi}
        {\em \cite{ahlswede1976spreading}}
        Let $\Theta,Z$ be binary random variables, and let $Z'$ be the output upon passing $Z$ over ${\rm BSC}(\frac{1-\delta}{2})$, with the Markov chain relation $\Theta \to Z \to Z'$. Then, we have
        \begin{equation}
            I(\Theta;Z') \leq \delta^2 I(\Theta; Z).
        \end{equation}
    \end{lemma}
    
    Fix a time index $i \ge 1$ and node index $j \ge 1$, and let $\vec{Y} = (Y_{1,j-1}, Y_{2,j-1}, \ldots, Y_{i-1,j-1}) \in \{0,1\}^{i-1}$ be the information received by node $j$ before time $i$.  Since $Y_{i,j-1}$ is the output of $X_{i,j-1}$ through ${\rm BSC}(\frac{1-\delta}{2})$, Lemma \ref{lem:sdpi} gives
    \begin{equation}
        I(\Theta; Y_{i,j-1} | \vec{Y}) \leq \delta^2 I(\Theta; X_{i,j-1}| \vec{Y}). \label{eq:apply_sdpi}
    \end{equation}
    Note that the Markov chain assumption in Lemma \ref{lem:sdpi} holds because conditioned on $\vec{Y}$, we have $\Theta \to X_{i,j-1} \to Y_{i,j-1}$ (due to the nodes' transmissions only depending on past received symbols).
    
    Define the function
    \begin{equation}
        f(i, j) = 
        \begin{cases} I(\Theta; (Y_{1,j-1}, Y_{2,j-1}, \ldots, Y_{i,j-1})),  & i\geq 1, j\geq 1\\
            1 & j=0\\
            0 & i=0, j\neq 0,
        \end{cases}
        \label{eq:f_def}
    \end{equation}
    which represents the amount of information node $j$ has received about $\Theta$ up to time $i$; the second and third cases hold because only node $0$ is given $\Theta$ before transmission starts.  We follow the high-level idea from \cite{schulman_1994} of setting up a recursion relation for $f$, but with rather different details.
    
    
    Considering the indices $i,j \geq 1$ that we used to define $\vec{Y}$, the chain rule for mutual information gives
    \begin{multline}
        f(i, j) = I(\Theta; \vec{Y}, Y_{i,j-1}) = I(\Theta; \vec{Y}) + I(\Theta; Y_{i, j-1}|\vec{Y}) \\ \leq f(i-1, j) + \delta^2 I(\Theta; X_{i, j-1}|\vec{Y}), \label{eq:f_rec}
    \end{multline}
    where the first term follows from the definition of $f$, and the second term follows from \eqref{eq:apply_sdpi}.
    
    To bound the second term in \eqref{eq:f_rec} in terms of $f$, we use another application of the chain rule:
    \begin{align}
        I(\Theta; X_{i,j-1}|\vec{Y}) &= I(\Theta;\vec{Y}, X_{i,j-1}) - I(\Theta;\vec{Y}) \nonumber \\ &\leq f(i-1, j-1) - f(i-1, j), \label{eq:to_sub}
    \end{align}
    where the term $f(i-1, j-1)$ arises differently depending on the value of $j$:
    \begin{itemize}
        \item If $j=1$, then $f(i-1, j-1) = 1$, which trivially upper bounds $I(\Theta;\vec{Y}, X_{i,j-1})$ as desired.
        \item If $j \ge 2$, then $f(i-1,j-1) = I(\Theta; (Y_{1,j-2}, Y_{2,j-2}, \ldots, Y_{i-1,j-2}))$; in this case, we use the data processing inequality, observing that we have the Markov chain relation $ \Theta \rightarrow (Y_{1,j-2}, Y_{2,j-2}, \ldots, Y_{i-1, j-2}) \rightarrow (\vec{Y}, X_{i, j-1})$.
    \end{itemize}
    Substituting \eqref{eq:to_sub} back into \eqref{eq:f_rec} gives
    \begin{equation} \label{prop_rate}
        f(i,j) \leq f(i-1, j) + \delta^2(f(i-1, j-1) - f(i-1, j)).
    \end{equation}
    Intuitively, this means that the propagation of information through time (represented by $f(i,j)-f(i-1,j)$) is only $\delta^2$ times as fast as the propagation of information through space (represented by $f(i-1, j-1) - f(i-1, j)$).  To formalize this intuition, we use the following technical lemma.
    
    \begin{lemma} \label{lem_gamma}
        Fix $\gamma > \delta^2$, and suppose that 
        \begin{equation} \label{eqn_c}
            e^{c\gamma} \geq 1 + \delta^2(e^c-1)
        \end{equation}
        for some $c>0$.  Then, we have for all $(i,j)$ that
        \begin{equation} \label{exp_bound}
            f(i,j) \leq e^{c(\gamma i - j)}.
        \end{equation}
    \end{lemma}
    \begin{proof}
        We proceed by induction.  When $i=0$ or $j=0$, it is immediate from the cases in \eqref{eq:f_def} that \eqref{exp_bound} is satisfied.  If $i,j \geq 1$, then by \eqref{prop_rate} and the induction hypothesis, we have 
        \begin{align}
            f(i,j) &\leq f(i-1, j) + \delta^2(f(i-1, j-1) - f(i-1, j)) \nonumber \\
            &\leq e^{c(\gamma (i-1) - j)} + \delta^2(e^{c(\gamma (i-1) - (j-1))} - e^{c(\gamma (i-1) - j)}), \label{eq:f_expanded}
        \end{align}
        where we note that the coefficients of $1$ and $-\delta^2$ combine to produce an overall positive coefficient to $f(i-1, j)$.
        
        Finally, substituting \eqref{eqn_c} into \eqref{eq:f_expanded} yields the desired upper bound on $f(i,j)$:
        \begin{multline}
            e^{c(\gamma (i-1) - j)} + \delta^2(e^{c(\gamma (i-1) - (j-1))} - e^{c(\gamma (i-1) - j)}) \\ = e^{c (\gamma (i-1) - j)}(1+\delta^2(e^c-1)) \leq e^{c (\gamma i-j)}.
        \end{multline}
    \end{proof}
    
    Observe that \eqref{eqn_c} is equivalent to
    \begin{equation}
        \frac{e^{c\gamma}-1}{e^c-1} \geq \delta^2,
    \end{equation}
    and since $\lim_{c \rightarrow 0} \frac{e^{c\gamma}-1}{e^c-1} = \gamma$, there always exists some $c>0$ satisfying this requirement when $\gamma > \delta^2$.  Fix any constant $v_0 > \delta^2$, and choose $\gamma \in (\delta^2,v_0)$, and a corresponding value of $c$ satisfying \eqref{eqn_c}.  We observe that 
    \begin{equation}
        \lim_{i\rightarrow \infty} f(i, \lfloor v_0\cdot i\rfloor) \leq \lim_{i\rightarrow \infty} e^{c(\gamma i - v_0 i)} = 0, \label{eq:f_v0}
    \end{equation}
    where the inequality follows from Lemma \ref{lem_gamma}, and the the limit of zero holds since $c > 0$ and $v_0 > \gamma$.
    
    We can now complete the proof of Proposition \ref{prop:converse}.  The quantity $f(i, \lfloor v_0\cdot i\rfloor)$ represents the mutual information between $\Theta$ and the information received by node $\lfloor v_0\cdot i\rfloor$ up to time $i$.  By the data processing inequality, this means that any estimate $\hat{\Theta} \in \{0,1\}$ formed using this information also has $I(\Theta;\hat{\Theta}) \to 0$, which implies that $\Pr(\Theta \ne \hat{\Theta}) \to \frac{1}{2}$.  Since $v_0$ is arbitrarily close to $\delta^2$ and represents the ratio of the node index to the time index in \eqref{eq:f_v0}, it follows that $v(\frac{1-\delta}{2}) \leq \delta^2$ as claimed.
    
    \section{Multi-Bit Variant} \label{sec:multi_bit}
    
    In this section, we provide a multi-bit variant of our protocol that can transmit a $k$-bit message across $m$ nodes with $\mathcal{O}(m+k)$ transmission time and asymptotically vanishing error probability.  We continue with the idea of arranging nodes into levels, but the arrangements are more complex, and we need to move beyond simple majority-of-3 coding.  For the latter purpose, the following technical lemma is useful.
    
    \begin{lemma} \label{lem:error_correcting}
        Suppose that we have a message consisting of $b \ge 1$ blocks, each having $k'$ bits.  Then there exists an encoding-decoding scheme such that when we transmit these $b$ blocks, followed by $\lceil \log_2 b \rceil + 1$ suitably-chosen additional blocks (also having $k'$ bits each), the original message can be recovered whenever at most 1 block is adversarially modified.  We will refer to the additional $\lceil \log_2 b \rceil + 1$ blocks as the redundancy blocks.
    \end{lemma}
    \begin{proof}
        The case $b=1$ is trivial, and when $b=2$ we can add a single redundancy block equal to the bit-wise XOR of the first two blocks.\footnote{When the lemma statement dictates us to add more blocks than necessary, we can simply add dummy blocks of all zeros.}  In the following, we assume that $b \ge 3$.
        
        When $b$ is of the form $2^r - r -1$ for some integer $r \ge 3$, we use a basic Hamming code with $r$ redundant bits \cite[Sec.~13.3]{mackay2003information}; it is applied bit-by-bit $k'$ times.  
        Otherwise, choose $r$ such that $2^{r-1}-(r-1)-1 < b < 2^r-r-1$. We temporarily extend the $b$ bits to $2^r-r-1$ bits by padding additional zeros, encode this new sequence using the Hamming code, and then delete the bits corresponding to the padded zeros (since they can trivially be re-inserted by the decoder).
        
        We need to show that the added redundancy for each of the $k'$ invocations of the Hamming code is at most $\lceil \log_2 b \rceil + 1$.  To see this, we first consider the case that $b = 2^r - r -1$.  The choice $r=3$ gives $b=4$, which satisfies the requirement.  On the other hand, when $r \ge 4$, we have $b = 2^r - r -1 > 2^{r-1}$, which implies $r < \log_2 b + 1$.  Since $b = 2^r - r -1$ cannot be a power of two for integer-valued $r \ge 4$, we deduce that $r \le \lceil \log_2 b \rceil$.
        
        It remains to handle general values of $b$.  Since $2^3 - 3 - 1 = 4$ and $2^4 - 4 - 1 = 11$, the choice $b=3$ rounds up to $r=3$, and the choices $b \in \{5,...,10\}$ round up to $r = 4$, so all of them satisfy $r \le \lceil \log_2 b \rceil + 1$.   For $r \ge 4$, when we incrementally increase $b$ between $2^r - r -1$ and $2^{r+1} - (r+1) - 1$, we pass exactly one power of two (namely, $2^r$).  Before that crossing, the rounding up (with respect to $r$) simply increases the above-mentioned bound $r \le \lceil \log_2 b \rceil$ (holding at the left endpoint) to $r \le \lceil \log_2 b \rceil + 1$, which gives the desired result.
    \end{proof} 
    
    \subsection{A Recursive Encoder and Decoder} \label{sec:enc_dec}
    
    In this subsection, we momentarily disregard the multi-hop nature of our problem, and describe a general recursive encoding-decoding strategy for a message of arbitrary length $k$.  This strategy provides a constant-rate code for reliable communication over the BSC relying only on elementary operations and the Hamming code subroutine (Lemma \ref{lem:error_correcting}), and may be of independent interest.

    \begin{figure*} 
        \begin{centering}
            \includegraphics[width=0.8\textwidth]{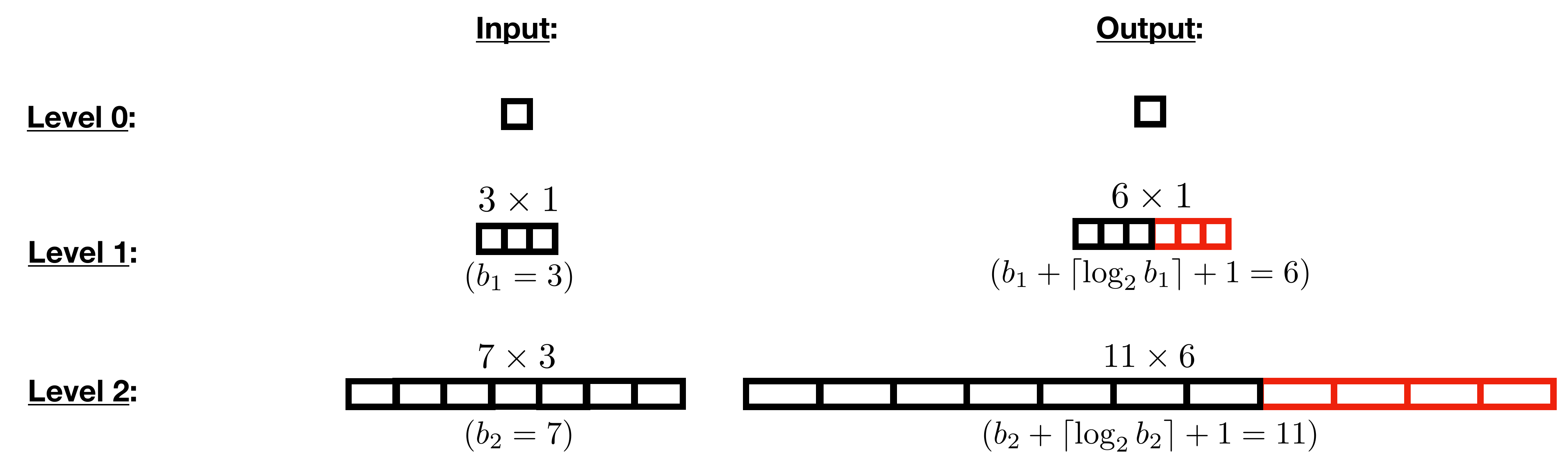}
            \par
        \end{centering}
        
        \caption{Illustration of our recursive multi-bit encoding scheme.  The additional blocks highlighted at the output correspond to added redundancy. \label{fig:MultiBit}}
    \end{figure*}
    
    Our strategy depends on a sequence $(b_1,b_2,\dotsc)$; this sequence is kept fully general for now.  
    The encoder is defined recursively according to levels indexed by $l=0,1,\dotsc,L$. The maximum level $L$ will be a function of $k$, and the levels $l=0,1,\dotsc,L-1$ will work with a smaller number of bits, denoted by $k_l < k$.  The encoders will also add redundancy, thus outputting some number $n_l$ of bits with $n_l \ge b_l$.  Specifically, we will have $k_0 = n_0 = 1$, and $k_l = \prod_{i=1}^l b_i$ and $n_l = \prod_{i=1}^l (b_i + \lceil \log_2 b_i \rceil + 1)$ for $l \ge 1$.
    
    The encoders are recursively defined as follows (see Figure \ref{fig:MultiBit} for an illustration):
    \begin{itemize}
        \item A level-$0$ encoder takes in 1 bit, and simply outputs that bit.  Hence, $k_0 = n_0 = 1$.
        \item For $l \ge 1$, a level-$l$ encoder takes in a block of $k_l = \prod_{i=1}^l b_i$ bits, breaks them up into $b_l$ blocks of length $k_{l-1} = \prod_{i=1}^{l-1} b_i$ each, and applies the level-$(l-1)$ encoding strategy to each block.  In addition, a further $\lceil \log_2 b_l \rceil + 1$ redundancy blocks are generated according to Lemma \ref{lem:error_correcting}, and they are also encoded using level-$(l-1)$ encoding.  The final output is a sequence of $n_l = \prod_{i=1}^l (b_i + \lceil \log_2 b_i \rceil + 1)$ bits.
    \end{itemize}
    The overall encoder of $k$ bits is not precisely any of those above, but makes use of the level-$L$ encoder, where $L$ is defined to be the value such that $k_L \le k < k_{L+1}$.  Specifically, the message of $k$ bits is divided into $\lceil k/k_L\rceil \le b_{L+1}$ blocks of $k_L$ bits each, padding the final block with zeros if necessary. The level-$L$ encoding strategy is then applied to each block. Note that padding the message at most doubles the message length.

    We similarly recursively define a sequence of decoders:
    \begin{itemize}
        \item A level-$0$ decoder simply reads the 1-bit input, and outputs it.
        \item For $l \ge 1$, a level-$l$ decoder reads $n_l = \prod_{i=1}^l (b_i + \lceil\log_2 b_i\rceil + 1)$ bits, and breaks them into $b_l + \lceil\log_2 b_l\rceil + 1$ blocks of size $n_{l-1} = \prod_{i=1}^{l-1} (b_i + \lceil\log_2 b_i\rceil + 1)$ each.  The level-($l-1$) decoding strategy is applied to each of these blocks (including the redundancy blocks).  Then, the resulting (decoded) blocks are further decoded according to Lemma \ref{lem:error_correcting}, on the assumption that at most one of the $b_l + \lceil\log_2 b_l\rceil + 1$ blocks has been corrupted.
    \end{itemize}
    The overall decoding strategy is to divide the $n$ received bits into $n_L$ blocks of size $n / n_L$ (divisibility holds by construction), and apply level-$L$ decoding on each one.

    \subsection{The Full Protocol}
    
    Our protocol depends not only on the sequence $(b_1,b_2,\dotsc)$ introduced above, but also on an additional sequence $(t_1,t_2,\dotsc)$.  It is useful to think of the 1-bit protocol from Section \ref{sec:protocol} as having $b_l = 3$ and $t_l = 4$ for all $l$ (though it does not append redundancy blocks).
    
    We first assign every node a level. Given $(t_1, t_2,\dotsc)$, the level of a node $j \in \{1,\dotsc,m-1\}$ is the largest integer $l$ such that $j$ is a multiple of $\prod_{i=1}^l t_i$ (we set $l=0$ if $j$ is not a multiple of $t_1$). As a result, between consecutive pairs of level-$l$ nodes, there are $t_l-1$ level-$(l-1)$ nodes. A level-$l$ node operates as follows:
    \begin{itemize}
        \item Read $n_l = \prod_{i=1}^l (b_i + \lceil\log_2 b_i\rceil + 1)$ bits;
        \item Decode these bits using the level-$l$ decoding strategy described above;
        \item Encode the result back using the level-$l$ encoding strategy;
        \item Transmit the resulting $n_l = \prod_{i=1}^l (b_i + \lceil\log_2 b_i\rceil + 1)$ bits to the next node.
    \end{itemize}
    This allocation produces some maximum level $L'$, which may differ from the maximum level $L$ in Section \ref{sec:enc_dec}.  In view of this difference, if $k < k_{L'}$, the encoder zero-pads the message to increase its length to $k_{L'}$; our analysis in Section \ref{sec:analysis} will be split according to whether this zero-padding is done or not.  Moreover, the decoder operates at level $\max\{L,L'\}$ regardless of divisibility by $\prod_{i=1}^l t_i$.
    
    
    As with our 1-bit protocol, the propagation of information is delayed, and each node ignores any ``non-specified'' bits before the first ``meaningful'' bit is received.
    
    We now state our main result for the multi-bit setting, which will be proved in Section \ref{sec:analysis}.
    
    \begin{theorem} \label{thm:main_multi}
        If $p < 3^{-8}/4$, then in the limit $m \to \infty$ (possibly with $k \to \infty$ simultaneously), the preceding protocol with $t_l = (l+2)^2$ and $b_l = \lfloor \frac{(l+2)^2}{4}\rfloor$ attains asymptotically vanishing error probability with a total transmission time of $n = \calO(m+k)$.
    \end{theorem}
    
    We state this result with respect to $m \to \infty$ for convenience, but one can easily adapt the protocol and its analysis to fix $m$ and only take $k \to \infty$; intuitively, adding more hops only makes the problem more difficult.  Moreover, in Corollary \ref{cor:main_multi} below, we remove the condition $p < 3^{-8}/4$ and give the precise scaling of the error probability.
    
    
    \subsection{Analysis of the Protocol} \label{sec:analysis}
    
    
    Here we prove Theorem \ref{thm:main_multi}. 
    Let $\epsilon_l$ denote the probability of failure associated with transmitting a length-$n_l$ block between two consecutive level-$l$ nodes (or from node 0 to the first one).  Each associated level-$(l-1)$ block of length $n_{l-1}$ has a failure probability of at most $t_l\cdot \epsilon_{l-1}$, by a union bound over the $t_l$ level-$(l-1)$ nodes the block passes through.  Hence, the probability of having two or more failures among the $b_l + \lceil \log_2 b_l \rceil + 1$ level-($l-1$) blocks is at most $\binom{b_l + \lceil \log_2 b_l \rceil + 1}{2} (t_{l}\epsilon_{l-1})^2$.  Since the redundancy added guarantees resilience to any single corrupted block, it follows that
    \begin{align}
        \epsilon_l &\leq \binom{b_l + \lceil \log_2 b_l \rceil + 1}{2} (t_{l}\epsilon_{l-1})^2 \nonumber \\ &\leq \frac12 ((b_l + \lceil \log_2 b_l \rceil + 1)t_l\epsilon_{l-1})^2.
    \end{align}
    It is convenient to note that $b_l + \lceil \log_2 b_l \rceil + 1 \le 2 b_l$, and accordingly simply this recursion to
    \begin{equation}
        \epsilon_l \le 2 (b_l t_l\epsilon_{l-1})^2.\label{eq:eps_recursion}
    \end{equation}
    
    \begin{proposition} \label{prop:suff_conds}
        Fix the sequences $(b_1,b_2,\dotsc)$ and $(t_1,t_2,\dotsc)$, and let $(\epsbar_1,\epsbar_2,\dotsc)$ be the sequence of upper bounds obtained via \eqref{eq:eps_recursion} (replacing inequality with equality) with $\epsilon_0 = p$.  Then, the preceding protocol transmits $k$ bits across $m$ hops with $n = \calO(m+k)$ transmission time and asymptotically vanishing error probability as $m \to \infty$ (possibly with $k \to \infty$ simultaneously), as long as the following conditions are satisfied: (i) $t_{l+1}b_{l+1}\epsbar_l \rightarrow 0$ as $l \to \infty$; (ii) $b_l/t_l \leq 1/2$ for all $l \ge 1$; and (iii) the summation $\sum_{l=1}^{\infty} \frac{\log b_l}{b_l}$ is finite.
    \end{proposition}
    
    We proceed with the proof of this proposition, referring to conditions (i)--(iii) throughout.
    
    We need to show that the total transmission time associated with a message of length $k$ over $m$ hops scales as $\calO(m+k)$.  Recall the notions of transmission delay and propagation delay from Section \ref{sec:1bit_analysis}.  
    We will show that the transmission delay scales as $\calO(m+k)$ and the propagation delay scales as $\calO(m)$, starting with the latter.
    
    A level-$l$ relay node waits to receive $\prod_{i=1}^l (b_i + \lceil\log_2 b_i\rceil + 1)$ bits before transmitting.  Hence, since there are at most $\frac{m}{\prod_{i=1}^l t_i}$ relay nodes at level $l$, the total propagation delay is at most
    \begin{align}
        &\sum_{l=1}^{L'}  \frac{m}{\prod_{i=1}^l t_i} \cdot \prod_{i=1}^l (b_i + \lceil\log_2 b_i\rceil + 1) \nonumber \\
            & \quad \le m \cdot \sum_{l=1}^{\infty} \bigg( \prod_{i=1}^l \frac{b_i}{t_i}\bigg(1 + \frac{3 \log_2 b_i}{b_i}\bigg) \bigg) \nonumber \\
            &\quad \leq c m\sum_{l=1}^{\infty} \bigg( \prod_{i=1}^l \frac{b_i}{t_i} \bigg), \label{eq:same_arg}
    \end{align}
    where $c$ is a constant such that $\prod_{i=1}^l (1+\frac{3\log_2 b_i}{b_i})\leq c$ for all $l$; the existence of such a constant is guaranteed because $\sum_{l=1}^{\infty} \frac{\log b_l}{b_l}$ is finite (condition (iii)).\footnote{This is most easily seen by applying the bound $1+z \le e^{z}$ to the term $1+\frac{3\log_2 b_i}{b_i}$.} Since $\frac{b_i}{t_i} \leq 1/2$ (condition (ii)), we also have that $\sum_{l=1}^{\infty} \prod_{i=1}^l \frac{b_i}{t_i}$ is finite, and hence, the propagation delay is linear in $m$. 
    
    To analyze the transmission delay, we consider two cases. 
    \begin{itemize}
        \item {\bf Case 1 ($k \leq k_{L'}$)}. Recall that in this case, the message is zero-padded to increase the length to $k_{L'}$.  Hence, a total of $\prod_{i=1}^{L'} (b_i + \lceil\log_2 b_i\rceil + 1)$ bits are sent by the encoder. Observe that this is equal to the propagation delay incurred by a level-$L'$ node, and we have already shown that this is at most $\calO(m)$.
        
        \item {\bf Case 2 ($k > k_{L'}$)}. In this case, the message length is not increased, and $L$ (from Section \ref{sec:enc_dec}) is the level at which the message is encoded (in $\lceil k/k_L \rceil \le b_{L+1}$ blocks). The multiplicative increase in length (of the message size) due to level-$L$ encoding takes the form
        \begin{equation}
            \prod_{l=1}^{L} \Big(1+\frac{\lceil \log_2 b_l \rceil + 1}{b_l}\Big), \label{eq:L_product}
        \end{equation}
        which is bounded independently of $L$ since $\sum_{l=1}^{\infty} \frac{\log b_l}{b_l}$ is finite (analogous to the last step in \eqref{eq:same_arg}).  Hence, the transmission delay is $\calO(k)$ in this case.
    \end{itemize}
    
    It remains to analyze the error probability.  We note that similarly to the analysis of the 1-bit protocol (Section \ref{sec:1bit_analysis}), even if the decoder's level $\max\{L,L'\}$ does not coincide with the notion of divisibility by $\prod_{i=1}^l t_i$ (corresponding to divisibility by $4^l$ in the 1-bit case), we can still perform the analysis as though the number of nodes were rounded up to produce such divisibility.   Then, using the same cases as above, we have the following:
    \begin{itemize}
        \item {\bf Case 1:} The message is encoded at level $L'$.  There are at most $b_{L'+1}$ level-$L'$ blocks, each of which passes through at most $t_{L'+1}$ level-$L'$ nodes.  The error probability associated with each of these is at most $\epsbar_{L'}$, so by the union bound, the overall error probability is at most $t_{L'+1} b_{L'+1} \epsbar_{L'}$.  By property (i) in Proposition \ref{prop:suff_conds}, this approaches zero as $L'$ increases.
        \item {\bf Case 2:} The message is encoded at level $L$ with $L \ge L'$, and the decoder is the only level-$L$ node.  There are again at most $b_{L+1}$ level-$L$ blocks, so similarly to Case 1, the overall error probability is at most $b_{L+1} \epsbar_L$.  Again using property (ii) in Proposition \ref{prop:suff_conds}, this approaches zero as $L$ increases (and since $L \ge L'$, it suffices that $L'$ increases).
    \end{itemize}
    The assumption $m \to \infty$ in Proposition \ref{prop:suff_conds} ensures that $L' \to \infty$, and we obtain asymptotically vanishing error probability as desired, thus proving Proposition \ref{prop:suff_conds}.  It remains to set $b_l$ and $t_l$ to prove Theorem \ref{thm:main_multi}.	
    
    {\bf Setting $b_l$ and $t_l$.} There are many possible choices of $b_l$ and $t_l$ such that the conditions of Proposition \ref{prop:suff_conds} are satisfied. 
    We consider the choices $t_l = (l+2)^2$ and $b_l = \lfloor \frac{(l+2)^2}{4}\rfloor$. Suppose there is a constant $c>1$ such that
    \begin{equation}
        \epsilon_0 \leq \frac{1}{4c} \cdot 3^{-8}. \label{eq:base_case}
    \end{equation}
    We first observe that substituting our choices of $b_l$ and $t_l$ into \eqref{eq:eps_recursion} gives
    \begin{equation}
        \epsilon_l \leq \frac{1}{8} (l+2)^8 \epsilon_{l-1}^2 \label{eq:eps_recursion2}
    \end{equation}
    We proceed to prove by induction that $\epsilon_l \leq \frac{c^{-2^l}}{4} (l+3)^{-8}$ for all $l \ge 1$.  The case $l=0$ holds by assumption in \eqref{eq:base_case}.  For the induction step, we substitute the induction hypothesis into \eqref{eq:eps_recursion2} to obtain
    \begin{align}
        \epsilon_l &\leq \frac{1}{8}(l+2)^8 ( (c^{-2^{l-1}}/4) \cdot (l+2)^{-8})^2 \nonumber \\ &=  \frac{1}{128} \cdot (l+2)^{-8} \cdot c^{-2^{l}} \leq \frac{1}{4} c^{-2^{l}}(l+3)^{-8},
    \end{align}
    where the last equality holds because $(l+3)^8 \leq 32 \cdot (l+2)^8$ for all $l\geq 0$; this is seen by noting that $\frac{l+3}{l+2} = 1+\frac{1}{l+2}$, which is maximized at $l=0$ and leads to $\big(\frac{3}{2}\big)^8 \approx 25.6 < 32$.
    
    Theorem \ref{thm:main_multi} now follows by recalling the assumption $p < 3^{-8}/4$ and setting $c = \frac{3^{-8}}{4p} > 1$.  Since $\epsilon_l \leq \frac{c^{-2^l}}{4} (l+3)^{-8}$ but $t_l$ and $b_l$ are only quadratic in $l$, we have the desired condition $t_{l+1}b_{l+1}\epsbar_l \rightarrow 0$ in Proposition \ref{prop:suff_conds}, with the left-hand side simplifying to $c^{-(2^l)(1+o(1))}$ (see Section \ref{sec:removing_assump} for a precise characterization of the final error probability).  The condition $b_l / t_l \le 1/2$ in Proposition \ref{prop:suff_conds} also trivially holds, and the final condition $\sum_{l=1}^{\infty} \frac{\log b_l}{b_l} < \infty$ holds because $b_l$ is quadratic in $l$.
    
    
    \subsection{Low-Noise and High-Noise Asymptotics}
    
    {\bf Low noise.} Similarly to the 1-bit setting, we handle the low-noise regime by increasing the number of level-0 nodes.  Here we do so by choosing $t_1 = \Theta(1/p)$ (keeping the higher $t_l$ values unchanged), with a suitably-chosen implied constant so that $\epsilon_1$ is kept below $3^{-8}/4$.  This means that a fraction $1 - \calO(p)$ of the nodes have level $0$ and are performing direct forwarding.  Thus, the contribution to the propagation delay is at most $m$ from level-$0$ nodes and $\calO(pm)$ from higher-level nodes (by a similar analysis to above), for a total of $m(1+\calO(p))$.
    
    Due to the term \eqref{eq:L_product} with $L$ that still grows unbounded for any fixed $p > 0$, an $\calO(k)$ term remains in the transmission delay, with an implied constant that does not approach one even as $p \to 0$.  Hence, the total transmission time is $n = \calO(k) + m(1+\calO(p))$.  When $k = o(m)$, this is optimal in view of the $m(1+\Omega(p))$ lower bound that holds even when $k=1$ (see Section \ref{sec:noise_asymp}).
    
    We do not attempt to establish the precise low-noise asymptotics for the remaining regimes $k = \Theta(m)$ and $m = o(k)$, which appear to be more challenging.  A Taylor expansion of the channel capacity gives a lower bound of $n \ge k\big(1 + \Omega\big(p\log\frac{1}{p}\big)\big)$ (even in the one-hop setting), but we leave it as an open problem as to when this can be matched in the multi-hop setting.
    
    {\bf High noise.} In the high-noise setting in which $p = \frac{1-\delta}{2}$, our protocol (with additional repetition coding as mentioned after Theorem \ref{thm:main_multi}) is optimal up to constant factors.  Since $\Theta(\delta^{-2})$ repetitions are required to reduce the effective crossover probability below $3^{-8}/4$, the transmission time scales as $\Theta(\delta^{-2}(m+k))$.
    
    In Section \ref{sec:noise_asymp}, we proved a lower bound of $n = \Omega(\delta^{-2} m)$ even for the 1-bit setting. In addition, since the capacity of ${\rm BSC}(\frac{1-\delta}{2})$ scales as $\Theta(\delta^{2})$ as $\delta \to 0$, we need $\Theta(\delta^{-2}k)$ channel uses to reliably send a message of length $k$, even in a 1-hop channel.  Therefore, the scaling $n = \Theta(\delta^{-2}(m+k))$ is optimal.
    
    \subsection{Variations and Discussion} \label{sec:removing_assump}
    
    We conclude this section by discussing the error probability, computation time, and an `anytime' version of our protocol.  To simplify the discussion, we focus on the case that $k = \Theta(m)$, and hence both also scale as $\Theta(n)$.
    
    {\bf Decay rate of the error probability.}  Recall that $k_l = \prod_{i=1}^l b_i$, so our choice $b_i = (i+2)^2$ gives $k_l = e^{\Theta(l \log l)}$.  A similar argument also gives $n_l = e^{\Theta(l \log l)}$, and it follows that both $L$ and $L'$ scale as $\Theta(\frac{\log n}{\log \log n})$ (e.g., to make $n_L$ coincide with $n$).  On the other hand, the error probability is dictated by $\epsilon_L = e^{-\Theta(2^L)}$ (or similarly with $L'$), and substituting $L = \Theta(\frac{\log n}{\log \log n})$ gives a final behavior of $\exp(-n^{\Theta(1/\log \log n)})$.   This is asymptotically vanishing with a decay rate marginally worse than $e^{-\Theta(n^c)}$ (for arbitrarily small $c$), but marginally better than $\frac{1}{{\rm poly}(n)}$ (with an arbitrarily large degree); these comparisons are are most easily seen by writing $\frac{1}{{\rm poly}(n)} = e^{-\Theta(\log n)}$ and comparing the terms in the exponents.
    
    The preceding discussion and the above idea of using $\Theta(\delta^{-2})$ repetitions are summarized in the following corollary.
    
    \begin{cor} \label{cor:main_multi}
        For any fixed $p \in \big(0,\frac{1}{2}\big)$, in the limit $m \to \infty$ with $k = \Theta(m)$, our protocol with $t_l = (l+2)^2$ and $b_l = \lfloor \frac{(l+2)^2}{4}\rfloor$, combined with a suitably-chosen constant number of repetitions, attains an error probability of at most $\exp(-n^{\Theta(1/\log \log n)})$ with a total transmission time of $n = \calO(m)$.
    \end{cor}
    
    {\bf Computation time.}
    Since Lemma \ref{lem:error_correcting} is based on the Hamming code, the encoding-decoding process associated with $b$ blocks, each having $k'$ bits, is given by $\calO(bk'\log b)$.  In particular, bit-wise Hamming decoding can be done by multiplying (modulo-2) a length-$b$ binary vector with the $b \times \calO(\log b)$ parity check matrix in $\calO(b \log b)$ time.  If one bit was corrupted, the resulting vector will index the corruption location.\footnote{In practice, the performance might be further improved by aggregating the $k'$ indices produced by the bit-wise decoding procedures.}
    
    For the encoder, the time spent to encode a single level-$l$ block is given by $\calO(b_l n_{l-1} \log b_l) = \calO(n_l \log b_l)$. Since there are a total of $\calO\big(\frac{n}{n_l}\big)$ such blocks, the total computation time spent on level-$l$ encoding is given by $\calO(n \log b_l)$.  Summing over all levels, the total computational complexity of encoding is as follows (or similarly with $L'$ in place of $L$):
    \begin{align}
        \sum_{l=0}^L \calO(n \log b_l) &= \calO\Big(n \log \Big(\prod_{l=0}^L b_l\Big)\Big) \nonumber \\  &= \calO(n \log k_L) = \calO(n \log n).
    \end{align}
    
    {\bf Encoding without knowledge of $k$.} When the encoder does not know $m$, we can consider a `level-$\infty$' encoder that reads an arbitrarily long message stream and recursively adapts its encoding as bits are received.  
    To achieve this, the encoder operates as follows: (i) Always send out each bit directly; (ii) Whenever the message length is a multiple of $\prod_{i=1}^l b_i$ for some $l \ge 1$, append the relevant redundancy blocks.   Since appending the redundancy blocks does not require modifying any of those already sent, the encoder does not need to wait for a block to arrive completely before sending according to what it has already received.  Thus, the desired `anytime' behavior of the encoder is achieved.
    
    \section{Conclusion}
    
    We have introduced a simple protocol for attaining positive information velocity, established order-optimal scaling laws for the low-noise and high-noise limits, and adapted our protocol to the multi-bit setting.  While several open questions remain, we believe that the following two directions are of particular interest:
    \begin{itemize}
        \item As we discussed in Section \ref{sec:1bit_var}, we did not optimize constant factors in our analysis.  It would be of significant interest to work towards tight upper and lower bounds on the information velocity via refined protocols and/or a more careful analysis.
        \item It would also be of interest to improve the error probability to be exponentially decaying with respect to the block length. The main result of \cite{schulman_1994} shows that this is information-theoretically possible (even in the multi-bit setting), but doing so with a simple protocol appears to be challenging.
    \end{itemize}
    
    \appendix
    
    \subsection{Overview of an Existing Approach} \label{sec:schulman}
    
    In this appendix, we present a simplified variant of Rajagopalan and Schulman's efficient protocol \cite[Sec.~6]{schulman_1994} that narrowly falls short of attaining positive information velocity.  It is useful to build this up from a weaker protocol.
    
    {\bf Protocol $\Psf_m^{(0)}$:} We first note that it is straightforward to transmit a single bit reliably over $m$ hops in $\calO(m \log m)$ channel uses.  This is because each node can estimate a bit repeated from the previous node in $\calO(\log m)$ channel uses with error probability $\frac{1}{m^2}$, and a union bound gives an overall error probability of at most $\frac{1}{m}$.
    
    For the purpose of building towards refined protocols, it is useful to improve this simple protocol to obtain a smaller error probability.  To achieve this, we can chain $m$ blocks of length $\calO(\log m)$ one after the other, and apply the above procedure independently in each block; each node transmits the estimate from its $(i-1)$-th block while receiving its $i$-th block.  The total transmission time remains $\calO(m \log m)$, but we claim that the upper bound on the error probability significantly improves to $e^{-m \log m}$.
    
    To see the latter claim, let the decoder take a majority vote over the $m$ blocks.  Since each block fails with probability at most $\frac{1}{m}$, the probability of any given set of $\frac{m}{2}$ blocks all failing is at most $\big(\frac{1}{m}\big)^{m/2} = e^{-(m/2) \log m}$.  Then, by a union bound over at most $2^m$ combinations of positions of the failing blocks, the overall error probability is at most $e^{-m \log m}$ for large enough $m$.

    {\bf Protocol $\Psf_m^{(1)}$:} Let $\Psf_{\log m}^{(0)}$ represent the preceding protocol with $\log m$ hops (ignoring rounding issues).  In this refined protocol, we use protocol $\Psf_{\log m}^{(0)}$ to transmit from node 0 to node $\log m$, then from node $\log m$ to node $2\log m +1$, and so on, until node $m$ is reached (ignoring divisibility issues).  Thus, there are at most $\frac{m}{\log m}$ sub-chains of nodes that apply protocol $\Psf_{\log m}^{(0)}$.
    
    Replacing $m$ by $\log m$ in the above $e^{-m \log m}$ error probability, we see that $\Psf_{\log m}^{(0)}$ fails with probability at most $e^{-\log m \cdot \log \log m}$.  Similarly to the above, we boost the error probability by repeating protocol $\Psf_{\log m}^{(0)}$ in successive blocks, this time doing so $\frac{m}{\log m}$ times.  The probability of $\frac{m}{ 2 \log m}$ specific blocks failing is at most $\big(e^{-\log m \cdot \log \log m}\big)^{m/(2 \log m)} = e^{-(m/2)\log\log m}$, and a union bound gives an overall error probability of at most $e^{-m \log \log m}$ for large enough $m$.  Moreover, the number of transmissions is now $\calO(m \log \log m)$, since the multiplicative overhead in applying protocol $\Psf_{\log m}^{(0)}$ is logarithmic with respect to the quantity in the subscript.
    
    {\bf Continuing recursively:} We can now replace $\Psf_{\log m}^{(0)}$ by $\Psf_{\log m}^{(1)}$ in the above reasoning, and continue recursively an arbitrary number of times, to obtain an error probability of at most $e^{-m \log\dotsc\log m}$ in $\calO(m \log \dotsc \log m)$ channel uses with an arbitrary number of logs.  Of course, the number of logs is only of theoretical (not practical) interest, since even $\log \log m$ is at most $5$ for any practical value of $m$.\footnote{Accordingly, the statements of the form ``for large enough $m$'' above may in fact amount to requiring enormous values.}  Nevertheless, no matter how many times we recurse, we fall narrowly short of attaining positive information velocity; this is due to the nodes waiting $\omega(1)$ time\footnote{Specifically, $\calO(\log n)$ time in protocol $\Psf_m^{(0)}$, $\calO(\log \log n)$ time in protocol $\Psf_m^{(1)}$, and so on.} to receive useful information from the previous node.
    
    {\bf Final protocol and result:} The protocol in \cite{schulman_1994} refines the above idea to allow for the transmission of $k$ bits, instead of just considering $k=1$, and accordingly also makes use of a constant-rate code that can correct a constant fraction of adversarial errors.  In addition, the parameters are set slightly differently to those above, though the same general idea (and limitation) remains.  The result in \cite[Thm.~6.1]{schulman_1994} states an error probability of $e^{-\Omega(m + k)}$ and a block length of $\calO((m+k)e^{O(\log^*(m+k))})$, where $\log^*(\cdot)$ is the iterated logarithm (i.e., the number of times $\log(\cdot)$ needs to be applied to bring the argument below one).  Note that $e^{O(\log^*(m+k))}$ is asymptotically smaller than any expression of the form $\log\dotsc\log(m+k)$.
    
    \subsection{Towards Exponentially Small Error Probability} \label{app:error_prob}
    
    In this appendix, we provide two approaches for modifying our 1-bit protocol to further reduce the error probability below $e^{-\Theta(\sqrt{n})}$, and as low as $e^{-\Theta(n^{1-\epsilon})}$ for any fixed $\epsilon > 0$.
    
    {\bf Approach 1: Modified constants.}  We modify the protocol in Section \ref{sec:protocol} to have two general integer-valued parameters, $b$ and $t$ (in the original protocol, $b=3$ and $t=4$):
    \begin{itemize}
        \item The level of a node $i \ge 1$ is the largest $l$ such that $t^l$ divides $i$.
        \item A level-$l$ node reads in $b^l$ bits, breaks it up into $b$ blocks of $b^{l-1}$ bits each, performs level-$(l-1)$ decoding on each, and estimates $\Theta$ according to a majority of $b$ before forwarding its estimate.
    \end{itemize}
    To avoid tie-breaking issues, we assume that $b$ is odd-valued.  The rest of the protocol remains the same.
    
    We again let $\epsilon_l$ represent the error probability associated with a level-$l$ block.  Instead of $\epsilon_{l+1} \le 48 \epsilon_l^2$ from the proof of Lemma \ref{lem:pe_recursion}, the recurrence relation for the error probability now becomes
    \begin{equation}
        \epsilon_{l+1} \leq \binom{b}{\frac{b+1}{2}} (t \epsilon_l)^{\frac{b+1}{2}}
    \end{equation}
    by an analogous argument as in Lemma \ref{lem:pe_recursion}.  Taking the log on both sides, we obtain
    \begin{equation}
        \log (\epsilon_{l+1}) \leq \frac{b+1}{2} \log (\epsilon_l) + O(1),
    \end{equation}
    where the $O(1)$ term depends on $b$ and $t$.  Then, for sufficiently small $\epsilon_0$ (a requirement we can overcome using repetition coding as usual), applying this equation recursively $l$ times gives
    \begin{equation}
        \log (\epsilon_{l+1}) = -\Theta\left(\left(\frac{b+1}{2}\right)^l\right). \label{eq:unrolled}
    \end{equation}
    On the other hand, up to rounding issues that only affect the constants, we have $m = t^L$, and substitution into \eqref{eq:unrolled} (with $l=L$) gives
    \begin{equation}
        \log (\epsilon_{L+1}) = -\Theta(m^{\log((b+1)/2)/\log t})
    \end{equation}
    To achieve positive information velocity, we need $t > b$.  It suffices to set $t = 2(b+1)$, and we observe that by choosing $b$ large enough, the resulting exponent $\frac{\log((b+1)/2)}{\log (2(b+1))}$ can be made arbitrarily close to one.  Thus, we deduce an error probability of $e^{-\Theta(m^{1-\epsilon})} = e^{-\Theta(n^{1-\epsilon})}$ for arbitrarily small $\epsilon > 0$ while maintaining $n = \calO(m)$, as desired.
    
    {\bf Approach 2: Chaining multiple instances.} In the proof of Lemma \ref{lem:pe_recursion}, we saw that the transmission delay is $3^{\lfloor \log_4 m \rfloor} = \calO(m^{\frac{\log 3}{\log 4}}) = o(m)$.  However, we can afford it to be $\calO(m)$ without affecting the requirement $n = \calO(m)$.  Hence, defining $\alpha = 1-\frac{\log 3}{\log 4} \approx 0.2075$, we can use the idea of repeating the entire protocol $m^{\alpha}$ times and decoding according to majority vote.  When doing so, we chain the instances of the protocol one after the other so that the nodes are continually transmitting after the initial delay incurred in the first instance; hence, the propagation delay is still only counted once, and only the transmission delay is blown up.
    
    The error probability for a single trial is $e^{-\Theta(\sqrt m)}$, so the probability of any particular set of $\frac{m^\alpha}{2}$ trials failing scales as $(e^{-\Theta(\sqrt m)})^{m^{\alpha}/2} = e^{-\Theta(m^{\frac{1}{2} + \alpha})}$.  A union bound over all ${m^{\alpha} \choose m^{\alpha}/2} \le 2^{m^{\alpha}}$ such sets contributes only a lower-order term in the exponent, so we still maintain $e^{-\Theta(m^{\frac{1}{2} + \alpha})} \approx e^{-\Theta(m^{0.7075})}$ scaling for the final error probability.
    
    We can build on this idea by using the parametrized protocol from Approach 1 above, but this time keeping $b=3$ and only increasing $t$.  For a single use of the protocol, the maximum number of levels is $L = \log_t m + O(1)$, and the power of $m$ in the exponent of the error probability is $\frac{1}{\log_2 t}$ (as a sanity check, setting $t = 4$ gives $\frac{1}{2}$).  This seemingly worsens as $t$ increases, but the advantage is that we can now repeat the entire protocol $m^{\alpha}$ times with $\alpha = 1 - \frac{\log 3}{\log t}$.  This leads to a final power of $m$ given by $\alpha + \frac{1}{\log_2 t} = 1 - \frac{\log_2 3 - 1}{\log_2 t}$, which can be made arbitrarily close to one by increasing $t$.
    
    {\bf Discussion.} In both of the above approaches, the ratio $\frac{m}{n}$ approaches zero as we take $\epsilon \to 0$ in the final $e^{-\Theta(n^{1-\epsilon})}$ scaling.  Hence, we narrowly fall short of attaining exponentially small error probability with $n = \calO(m)$, which is known to be information-theoretically possible \cite{schulman_1994}.
    
    \bibliographystyle{IEEEtran}
    \bibliography{general}
    
     \begin{IEEEbiographynophoto}{Yan Hao Ling}
    received the B.Comp.~degree in computer science and the B.Sci.~degree 
    in mathematics from the National University of Singapore (NUS) in 2021. 
    He is now a PhD student in the Department of Computer Science at NUS.
    His research interests are in the areas of
    information theory, statistical learning, and theoretical computer science.
\end{IEEEbiographynophoto}

\begin{IEEEbiographynophoto}{Jonathan Scarlett}
    (S'14 -- M'15) received 
    the B.Eng.~degree in electrical engineering and the B.Sci.~degree in 
    computer science from the University of Melbourne, Australia. 
    From October 2011 to August 2014, he
    was a Ph.D. student in the Signal Processing and Communications Group
    at the University of Cambridge, United Kingdom. From September 2014 to
    September 2017, he was post-doctoral researcher with the Laboratory for
    Information and Inference Systems at the \'Ecole Polytechnique F\'ed\'erale
    de Lausanne, Switzerland. Since January 2018, he has been an assistant
    professor in the Department of Computer Science and Department of Mathematics,
    National University of Singapore. His research interests are in
    the areas of information theory, machine learning, signal processing, and
    high-dimensional statistics. He received the Singapore National Research Foundation (NRF)
    fellowship, and the NUS Early Career Research Award.
\end{IEEEbiographynophoto}
    
\end{document}